\documentclass[10pt,twocolumn]{article}

\usepackage{amsmath,amssymb,amsthm,bm}
\usepackage{graphicx}
\usepackage{xcolor}
\usepackage{booktabs}
\usepackage{microtype}
\usepackage[margin=1.65cm,columnsep=0.65cm]{geometry}
\usepackage[numbers,sort&compress]{natbib}
\usepackage[hidelinks]{hyperref}
\graphicspath{{figures/}}

\newtheorem{theorem}{Theorem}
\newtheorem{proposition}{Proposition}
\newtheorem{corollary}{Corollary}

\newcommand{\dd}{\mathrm{d}}
\newcommand{\e}{\mathrm{e}}
\newcommand{\ii}{\mathrm{i}}
\newcommand{\E}{\mathbb{E}}
\newcommand{\Var}{\operatorname{Var}}
\newcommand{\wrap}{\operatorname{wrap}}

\begin{document}

\title{Operational Emergence of a Global Phase under Time-Dependent Coupling 
in Oscillator Networks}

\author{Ver\'onica Sanz\\[2pt]
\small Theoretical Physics Department, Universitat de Val\`encia, Spain\\
\small Instituto de F\'isica Corpuscular (IFIC), CSIC--Universitat de Val\`encia, Spain}

\date{\today}

\maketitle

\begin{abstract}
Time-dependent coupling is often interpreted as introducing competition between a protocol rate and an intrinsic synchronization rate.  We show that this interpretation is unavailable in the identical, overdamped Kuramoto model when time dependence enters only through a scalar multiplier of a fixed coupling field.  The accumulated coupling $S(t)=\int_0^tK(u)\dd u$ is then an exact dynamical clock: equal-exposure protocols traverse the same autonomous orbit, and integrable decays produce finite-exposure arrest rather than deterministic loss of adiabatic tracking.

This clock suggests a classification under time reparametrization.  Deterministic orbit observables and winding sectors are invariant; an additive perturbation $\varepsilon\bm G$ has a leading response that is a linear functional of the reciprocal schedule $w(s)=1/K(t(s))$ against a kernel fixed by the autonomous orbit; and inertia or noncommuting graph generators lie outside this reciprocal-weighted class.  Laboratory-time white noise has the corresponding $w$-weighted covariance.  We derive exact mode covariances, prove a universal terminal bang--bang optimum under bounded coupling and fixed exposure, and validate both results in nonlinear networks.

We also separate information in the full phase sample from information retained by the order parameter.  Conditional observation error of $\arg Z$ follows a $1/(NR^2)$ law with a configuration-dependent prefactor, whereas a known-template experiment has Fisher information $N/\sigma^2$ independently of $R$; the efficiency of $\arg Z$ is $R^2/q_2$.  Finally, periodic rings exhibit stable winding sectors with strong local but vanishing global order.  Conventional Kibble--Zurek freeze-out requires a $K$-dependent critical structure and is therefore absent in the exact-clock sector; the ingredients that create such a transition also break the reduction.
\end{abstract}

\section{Introduction}

The Kuramoto model provides a canonical description of collective phase ordering in physical, biological, and engineered systems~\cite{Kuramoto1984,Acebron2005,DorflerBullo2014}.  On a fixed undirected graph with adjacency matrix $A$, identical overdamped oscillators obey
\begin{equation}
 \dot\theta_i=K\sum_{j=1}^{N}A_{ij}\sin(\theta_j-\theta_i),
 \label{eq:autonomous-kuramoto}
\end{equation}
and their global coherence is summarized by
\begin{equation}
 Z=R\e^{\ii\Psi}:=\frac{1}{N}\sum_{j=1}^{N}\e^{\ii\theta_j}.
 \label{eq:order-parameter}
\end{equation}
Time-dependent frequencies, couplings, and graph structures arise naturally in open and controlled oscillator systems~\cite{Petkoski2012,PietrasDaffertshofer2016,Faggian2019}.  In such general settings, competition between protocol and relaxation times can be genuine: oscillator-specific forcing changes the vector field~\cite{Petkoski2012}, network rewiring competes with local synchronization and desynchronization~\cite{Faggian2019}, and quenches through a symmetry-breaking transition can freeze topological defects~\cite{Miranda2013}.  None of these mechanisms is equivalent to multiplying a fixed deterministic vector field by a scalar function of time.

The restricted scalar case therefore requires separate treatment.  A changing parameter produces a genuine tracking problem only if it changes the vector field's orbits, equilibria, invariant measure, or stability structure.  In the identical noiseless overdamped model, a scalar $K(t)$ does none of these: it multiplies the complete vector field.  Its only deterministic effect is to change the speed at which a fixed autonomous orbit is traversed.  This is an elementary instance of time reparametrization, or orbital equivalence~\cite{GelfertMotter2010}; the contribution here is not the reparametrization itself, but its systematic consequences and the diagnostic programme it implies for synchronization protocols.

The purpose of this paper is to establish the relevant boundary precisely through a single organizing principle: quantities are classified by their transformation under the exposure clock.  The first class is invariant and includes deterministic orbits, spectral exposures, and deterministic winding-sector selection.  The second is reciprocal weighted: constant frequencies, pinning, external drift, and noise become sensitive to how the fixed laboratory-time mass is distributed along the autonomous orbit.  The third lies outside scalar reparametrization, as occurs for inertia, explicitly time-dependent forcing, or noncommuting graph generators.

We derive and test representative results in each class.  Besides exact clock equivalence, we obtain a general first-order response kernel for additive perturbations, an exact covariance and bounded-control optimum for noisy modes, and a statistical efficiency result showing precisely how compression to $Z$ discards phase information.  Periodic rings then isolate topology as an invariant obstruction to global coherence.  This classification also identifies the boundary of conventional Kibble--Zurek reasoning: a convincing claim of critical freeze-out must first exhibit the clock-breaking ingredient that creates a $K$-dependent critical structure.

\section{Deterministic scalar protocols}
\label{sec:deterministic}

\subsection{Model and gradient structure}

Let $G=(V,E)$ be a fixed connected undirected weighted graph, with $A_{ij}=A_{ji}\geq0$ and $A_{ii}=0$.  We consider
\begin{equation}
 \dot{\bm\theta}(t)=K(t)\bm F(\bm\theta(t)),
 \qquad
 F_i(\bm\theta)=\sum_jA_{ij}\sin(\theta_j-\theta_i),
 \label{eq:scalar-protocol}
\end{equation}
on the $N$-torus.  For the clock reduction, $K$ may be any locally integrable signed function; the Kuramoto field is complete because the torus is compact.  The monotone-exposure interpretation, energy dissipation, and all protocols used below assume $K\geq0$.  Define the accumulated coupling exposure
\begin{equation}
 S(t):=\int_0^tK(u)\dd u.
 \label{eq:exposure}
\end{equation}
The autonomous unit-coupling system is
\begin{equation}
 \frac{\dd\bm\vartheta}{\dd s}=\bm F(\bm\vartheta),
 \qquad \bm\vartheta(0)=\bm\theta_0,
 \label{eq:unit-flow}
\end{equation}
and its flow is denoted by $\Phi_s(\bm\theta_0)$.

The Kuramoto field is the negative gradient of
\begin{equation}
 U(\bm\theta)=-\sum_{i<j}A_{ij}\cos(\theta_i-\theta_j).
 \label{eq:potential}
\end{equation}
Consequently, when $K(t)\geq0$,
\begin{equation}
 \frac{\dd U}{\dd t}=-K(t)\|\nabla U\|^2\leq0,
 \qquad
 \frac{\dd U}{\dd s}=-\|\nabla U\|^2.
 \label{eq:gradient-dissipation}
\end{equation}
Changing $K(t)$ changes the rate of descent in physical time, but not the energy landscape or the orbit in configuration space.

\subsection{Exact clock equivalence}

\begin{theorem}[Clock equivalence]
\label{thm:clock}
Let $\bm F$ be locally Lipschitz and generate a complete flow, and let $K(t)$ be locally integrable, with either sign.  The unique solution of Eq.~\eqref{eq:scalar-protocol} is
\begin{equation}
 \boxed{\bm\theta(t)=\Phi_{S(t)}(\bm\theta_0)}.
 \label{eq:clock-equivalence}
\end{equation}
In particular, if two protocols $K_a$ and $K_b$ satisfy
$S_a(t_a)=S_b(t_b)$, then
\begin{equation}
 \bm\theta_a(t_a)=\bm\theta_b(t_b)
 \label{eq:equal-exposure}
\end{equation}
for the same initial state.
\end{theorem}

\begin{proof}
Set $\bm x(t)=\Phi_{S(t)}(\bm\theta_0)$.  Since $S$ is absolutely continuous, the chain rule holds almost everywhere and gives
\begin{equation*}
 \dot{\bm x}(t)=\dot S(t)\,\partial_s\Phi_s(\bm\theta_0)|_{s=S(t)}
 =K(t)\bm F(\bm x(t)).
\end{equation*}
Moreover $\bm x(0)=\bm\theta_0$.  Uniqueness identifies $\bm x$ with the solution of Eq.~\eqref{eq:scalar-protocol}.  The equal-exposure statement follows immediately.
\end{proof}

\begin{corollary}[Finite-exposure arrest]
\label{cor:finite-exposure}
If $K\geq0$ and $S_\infty:=\int_0^\infty K(t)\dd t<\infty$, then
\begin{equation}
 \lim_{t\to\infty}\bm\theta(t)=\Phi_{S_\infty}(\bm\theta_0).
 \label{eq:finite-exposure-limit}
\end{equation}
Thus the protocol stops the system at a finite autonomous time.  This is finite-exposure arrest, not failure to follow a moving equilibrium.
\end{corollary}

For the algebraic decay
\begin{equation}
 K(t)=K_0(1+t/\tau)^{-\alpha},
 \label{eq:algebraic-protocol}
\end{equation}
one finds
\begin{equation}
 S_\infty=
 \begin{cases}
  K_0\tau/(\alpha-1),&\alpha>1,\\
  \infty,&\alpha\leq1.
 \end{cases}
 \label{eq:algebraic-exposure}
\end{equation}
Thus, at fixed $K_0$ and $\alpha>1$, increasing $\tau$ increases the total interaction budget linearly.  A collapse against $\lambda_2\tau$ in this one-parameter family can therefore be an accumulated-exposure collapse, because $\lambda_2S_\infty=K_0\lambda_2\tau/(\alpha-1)$.

\subsection{Laplacian modes and spectral exposure}

Let $L=D-A$ be the combinatorial graph Laplacian, with
$0=\lambda_1<\lambda_2\leq\cdots\leq\lambda_N$~\cite{Restrepo2005,DorflerBullo2014}.  Near a coherent state, write
$\bm\theta=\bar\theta\bm1+\bm\varphi$ with
$\bm1^\mathsf{T}\bm\varphi=0$.  Then
\begin{equation}
 \dot{\bm\varphi}=-K(t)L\bm\varphi
 +O(\|B^\mathsf{T}\bm\varphi\|^3),
 \label{eq:linearization}
\end{equation}
where $B$ is an oriented incidence matrix.  In the Laplacian eigenbasis,
\begin{equation}
 \dot c_\alpha=-K(t)\lambda_\alpha c_\alpha,
 \qquad
 c_\alpha(t)=c_\alpha(0)\e^{-\lambda_\alpha S(t)}.
 \label{eq:mode-solution}
\end{equation}
The natural variables are therefore the spectral exposures
$x_\alpha=\lambda_\alpha S$, not a comparison with $|\partial_t\ln K|$.

For the quadratic disagreement $E=\|\bm\varphi\|^2$, the linear result is
\begin{equation}
 E(t)=\sum_{\alpha=2}^{N}|c_\alpha(0)|^2
 \e^{-2\lambda_\alpha S(t)}.
 \label{eq:energy-spectrum}
\end{equation}
Only after faster modes have decayed does $\lambda_2$ alone control $E$.  This also shows why
\begin{equation}
 \frac{-\frac12\partial_t\ln E}{K(t)\lambda_2}
 \label{eq:ratio-warning}
\end{equation}
cannot diagnose deterministic nonadiabaticity: for a pure slow mode it equals one for every scalar protocol, while for a mode mixture it is a weighted spectral average greater than or equal to one within the linear regime.

\section{Clock-breaking perturbations and reciprocal schedules}
\label{sec:noise}

\subsection{Universal first-order response kernel}

Consider a time-independent additive perturbation that does not share the scalar multiplier,
\begin{equation}
 \dot{\bm\theta}_\varepsilon
 =K(t)\bm F(\bm\theta_\varepsilon)
 +\varepsilon\bm G(\bm\theta_\varepsilon).
 \label{eq:additive-perturbation}
\end{equation}
This includes weak frequency heterogeneity, constant pinning, and autonomous external drift.

\begin{theorem}[Reciprocal-schedule response]
\label{thm:reciprocal-kernel}
Let $\bm F$ and $\bm G$ be continuously differentiable and let $K(t)\geq K_{\min}>0$ on $[0,T]$.  Write $S_T=S(T)$, $\bm\vartheta(s)=\Phi_s(\bm\theta_0)$, and
\begin{equation}
 \bm\theta_\varepsilon(t(s))
 =\bm\vartheta(s)+\varepsilon\bm\psi(s)+O(\varepsilon^2).
 \label{eq:response-expansion}
\end{equation}
If $M(s,u)$ is the fundamental matrix generated by
$A(s)=D\bm F(\bm\vartheta(s))$, then
\begin{equation}
 \begin{aligned}
 \bm\psi(S_T)&=\int_0^{S_T}
 M(S_T,s)\bm G(\bm\vartheta(s))w(s)\dd s,\\
 w(s)&=\frac{1}{K(t(s))}.
 \end{aligned}
 \label{eq:reciprocal-kernel}
\end{equation}
At fixed physical duration,
\begin{equation}
 \int_0^{S_T}w(s)\dd s=T.
 \label{eq:reciprocal-mass}
\end{equation}
\end{theorem}

\begin{proof}
Changing variables in Eq.~\eqref{eq:additive-perturbation} gives
\begin{equation*}
 \frac{\dd\bm\theta_\varepsilon}{\dd s}
 =\bm F(\bm\theta_\varepsilon)
 +\varepsilon w(s)\bm G(\bm\theta_\varepsilon).
\end{equation*}
Differentiation with respect to $\varepsilon$ at zero yields
$\bm\psi'=A(s)\bm\psi+w(s)\bm G(\bm\vartheta(s))$ and
$\bm\psi(0)=0$.  Variation of constants gives
Eq.~\eqref{eq:reciprocal-kernel}.  Finally,
$w(s)\dd s=K(t(s))^{-1}K(t)\dd t=\dd t$, proving
Eq.~\eqref{eq:reciprocal-mass}.
\end{proof}

Thus, at first order, the autonomous orbit fixes a universal vector-valued response kernel and the schedule only redistributes a positive weight of fixed mass along it.  This statement alone does not order arbitrary vector responses, because the kernel may change direction or sign.  An ordering follows when the relevant scalar or matrix kernel is positive, as in the covariance problem below.

\subsection{Stochastic time change}

Consider additive independent noise in laboratory time~\cite{Sonnenschein2013},
\begin{equation}
 \dd\bm\theta
 =K(t)\bm F(\bm\theta)\dd t+\sqrt{2D}\,\dd\bm W_t.
 \label{eq:stochastic-model}
\end{equation}
On an interval where $K(t)>0$, invert $s=S(t)$ and write $t=t(s)$.

\begin{proposition}[Noise in the interaction clock]
\label{prop:stochastic-clock}
If $K(t)>0$ almost everywhere on the interval considered, then $S$ is strictly increasing and, under the stochastic time change $s=S(t)$, Eq.~\eqref{eq:stochastic-model} becomes
\begin{equation}
 \dd\bm\theta
 =\bm F(\bm\theta)\dd s
 +\sqrt{\frac{2D}{K(t(s))}}\,\dd\bm B_s,
 \label{eq:stochastic-timechange}
\end{equation}
where $\bm B_s$ is a standard Brownian motion in the $s$ clock.
\end{proposition}

Thus the deterministic orbit equivalence survives in the drift but fixed laboratory-time noise becomes a schedule-dependent effective temperature $D/K$.  Equal total exposure is no longer sufficient to determine the distribution of the final state.

Linearizing along $\bm\vartheta(s)$ gives
$\dd\bm\psi=A(s)\bm\psi\dd s+\sqrt{2Dw(s)}\dd\bm B_s$.
Consequently, its covariance at $S_T$ is
\begin{align}
 C(S_T)&=M(S_T,0)C(0)M(S_T,0)^\mathsf{T}\notag\\
 &\quad+2D\int_0^{S_T}w(s)M(S_T,s)M(S_T,s)^\mathsf{T}\dd s.
 \label{eq:general-covariance-kernel}
\end{align}
White noise therefore carries the same reciprocal weight in its covariance, although its pathwise amplitude transforms as $K^{-1/2}$.

\subsection{Exact covariance of a Laplacian mode}

In the linearized regime, each nonzero mode obeys
\begin{equation}
 \dd c_\alpha=-\lambda_\alpha K(t)c_\alpha\dd t
 +\sqrt{2D}\,\dd W_{\alpha,t}.
 \label{eq:ou-mode}
\end{equation}

\begin{proposition}[Mode mean and variance]
\label{prop:ou}
Let $m_\alpha(t)=\E[c_\alpha(t)]$ and
$v_\alpha(t)=\Var[c_\alpha(t)]$.  Then
\begin{align}
 m_\alpha(t)&=m_\alpha(0)\e^{-\lambda_\alpha S(t)},
 \label{eq:ou-mean}\\
 v_\alpha(t)&=v_\alpha(0)\e^{-2\lambda_\alpha S(t)}
 +2D\int_0^t\e^{-2\lambda_\alpha[S(t)-S(u)]}\dd u.
 \label{eq:ou-variance}
\end{align}
\end{proposition}

\begin{proof}
Multiplication of Eq.~\eqref{eq:ou-mode} by the integrating factor
$\e^{\lambda_\alpha S(t)}$ gives the explicit It\^o solution.  Taking the expectation gives Eq.~\eqref{eq:ou-mean}; It\^o isometry gives Eq.~\eqref{eq:ou-variance}.
\end{proof}

The mean retains exact exposure equivalence.  The second term in
Eq.~\eqref{eq:ou-variance}, however, depends on when the coupling is applied.  At fixed terminal exposure $S_T$, a schedule with smaller partial exposure $S(u)$ has larger surviving exposure $S_T-S(u)$ and therefore suppresses noise injected at every time $u$ more strongly.  In particular, if
$S_{
m B}(u)\leq S_0(u)\leq S_{
m F}(u)$ pointwise and the three schedules share the same endpoint, then
$v_{\rm B}(T)\leq v_0(T)\leq v_{\rm F}(T)$.

For every nonnegative schedule with fixed $T$ and $S_T$, Eq.~\eqref{eq:ou-variance} also gives the sharp limiting bounds
\begin{equation}
 \bigl[v_\alpha(0)+2DT\bigr]\e^{-2\lambda_\alpha S_T}
 \leq v_\alpha(T)\leq
 v_\alpha(0)\e^{-2\lambda_\alpha S_T}+2DT.
 \label{eq:variance-bounds}
\end{equation}
The lower and upper limits are approached by concentrating the available coupling in a terminal or initial pulse, respectively.  Thus schedule dependence at fixed exposure can be parametrically large.

\begin{proposition}[Universal bounded optimum]
\label{prop:bounded-optimum}
Fix $T$, $S_T\leq K_{\max}T$, and the admissible class
$0\leq K(t)\leq K_{\max}$ with $\int_0^T K(t)\dd t=S_T$.  Every nonzero Laplacian-mode variance is minimized by the terminal bang--bang schedule
\begin{equation}
 K_*(t)=
 \begin{cases}
 0,&0\leq t<T-S_T/K_{\max},\\
 K_{\max},&T-S_T/K_{\max}<t\leq T.
 \end{cases}
 \label{eq:terminal-bangbang}
\end{equation}
The same schedule minimizes any nonnegative weighted sum of the mode variances.
\end{proposition}

\begin{proof}
Every admissible schedule must satisfy
\begin{equation*}
 S(u)\geq\max\{0,S_T-K_{\max}(T-u)\}=:S_*(u),
\end{equation*}
because at most $K_{\max}(T-u)$ exposure can be supplied after $u$.
The schedule in Eq.~\eqref{eq:terminal-bangbang} attains this lower envelope.  At fixed $S_T$, the noise kernel in Eq.~\eqref{eq:ou-variance} is pointwise increasing in $S(u)$, so $S_*$ minimizes every $v_\alpha(T)$ simultaneously.
\end{proof}

Without an amplitude bound, no ordinary schedule attains the lower limit in Eq.~\eqref{eq:variance-bounds}; it is the infimum approached by a terminal pulse.  Alternative resource constraints, such as fixed $\int K^2\dd t$, allow $S_T$ to vary and lead to a genuinely mode- and graph-dependent optimal-control problem, analogous in structure but not identical in cost to finite-time optimization in stochastic thermodynamics~\cite{SchmiedlSeifert2007}.

At fixed $K$, the stationary variance is $D/(K\lambda_\alpha)$.  If $K$ changes slowly, one may compare the relaxation rate $2K\lambda_\alpha$ with $|\partial_t\ln K|$ to study tracking of this moving stochastic covariance.  Such a criterion concerns the noisy invariant distribution, not deterministic tracking of an equilibrium whose location is independent of $K$.

\subsection{Beyond reciprocal weighting}

The reciprocal kernel covers autonomous additive terms such as
\begin{align}
 \dot\theta_i&=\omega_i+K(t)F_i(\bm\theta),
 \label{eq:heterogeneous}\\
 \dot{\bm\theta}&=K(t)\bm F(\bm\theta)+\varepsilon\bm G(\bm\theta).
 \label{eq:additive-class}
\end{align}
At first order, their schedule dependence is completely described by Eq.~\eqref{eq:reciprocal-kernel}.  Structurally different examples include
\begin{align}
 m\ddot\theta_i+\gamma\dot\theta_i&=K(t)F_i(\bm\theta),
 \label{eq:inertial}\\
 \dot{\bm\theta}&=-K(t)L(t)\bm\theta,
 \label{eq:changing-graph}
\end{align}
as well as explicitly timed forcing and delays.  Inertia generates additional derivative terms under a time change, while noncommuting graph generators cannot be reduced to one accumulated scalar.  These are appropriate settings for genuine rate-induced transitions or tracking failure~\cite{Ashwin2012}.

\section{Operational phase estimability}
\label{sec:phase}

The phase $\Psi=\arg Z$ is a circular quantity.  An ordinary ensemble variance of $\Psi$ is neither gauge invariant across independently rotated realizations nor well behaved near the branch cut.  Operational estimability must specify what is observed, what is held fixed, and which circular error is used~\cite{MardiaJupp2000}.

\subsection{A conditional observation model}

Fix a microscopic configuration $\bm\theta$ with $R>0$ and suppose its phases are observed with independent small errors,
\begin{equation}
 \widetilde\theta_j=\theta_j+\epsilon_j,
 \qquad
 \epsilon_j=\sigma\xi_j,
 \qquad
 \E\xi_j=0,
 \qquad
 \E\xi_i\xi_j=\delta_{ij}.
 \label{eq:observation-model}
\end{equation}
The $\xi_j$ are taken i.i.d., symmetric, and with finite sixth moment.  These assumptions include the Gaussian errors used below; symmetry excludes a common directional bias from the observation channel.
Let $\widehat\Psi=\arg[N^{-1}\sum_j\e^{\ii\widetilde\theta_j}]$ and define the local circular error
$\delta\Psi=\wrap(\widehat\Psi-\Psi)$.

\begin{proposition}[Conditional phase error]
\label{prop:phase-error}
Put $\phi_j=\theta_j-\Psi$, $c_j=\cos\phi_j$, and $s_j=\sin\phi_j$.  In the joint small-error regime $\sigma\ll1$ and $\sigma/(R\sqrt N)\ll1$,
\begin{align}
 \delta\Psi
 &=\frac{1}{NR}\sum_jc_j\epsilon_j
 -\frac{1}{2NR}\sum_js_j(\epsilon_j^2-\sigma^2)\notag\\
 &\quad+\frac{1}{N^2R^2}
 \left(\sum_jc_j\epsilon_j\right)
 \left(\sum_ks_k\epsilon_k\right)+\mathcal R_3,
 \label{eq:phase-delta}\\
 \mathcal R_3
 &=O_{\mathbb P}\!\left(
 \frac{\sigma^3}{R\sqrt N}
 +\frac{\sigma^3}{NR^2}
 +\frac{\sigma^3}{N^{3/2}R^3}\right),
 \label{eq:phase-remainder}\\
 \E[(\delta\Psi)^2\mid\bm\theta]
 &=\frac{\sigma^2}{NR^2}q_2(\bm\theta)
 +O\!\left(\frac{\sigma^4}{NR^2}
 +\frac{\sigma^4}{N^2R^4}\right),
 \label{eq:phase-mse}\\
 q_2(\bm\theta)&:=\frac1N\sum_{j=1}^{N}
 c_j^2.
 \label{eq:q2}
\end{align}
Moreover, with $Z_2=N^{-1}\sum_j\e^{2\ii\theta_j}$, the estimator has the second-order conditional bias
\begin{equation}
 \E[\delta\Psi\mid\bm\theta]
 =\frac{\sigma^2}{2NR^2}
 \operatorname{Im}\!\left(\e^{-2\ii\Psi}Z_2\right)
 +O\!\left(\frac{\sigma^4}{NR^2}
 +\frac{\sigma^4}{N^2R^4}\right).
 \label{eq:phase-bias}
\end{equation}
\end{proposition}

\begin{proof}
Rotate by $\Psi$ and write the perturbed order parameter as
$\widetilde Z\e^{-\ii\Psi}=R+x+\ii y$.  Since
$\sum_js_j=0$ and $\sum_jc_j=NR$, expansion of
$\sin(\phi_j+\epsilon_j)$ and $\cos(\phi_j+\epsilon_j)$ gives
\begin{align*}
 y&=\frac1N\sum_jc_j\epsilon_j
 -\frac1{2N}\sum_js_j(\epsilon_j^2-\sigma^2)+O_{\mathbb P}(\sigma^3/\sqrt N),\\
 x&=-\frac1N\sum_js_j\epsilon_j-\frac{\sigma^2R}{2}
 +O_{\mathbb P}(\sigma^2/\sqrt N).
\end{align*}
Using $\delta\Psi=\arctan[y/(R+x)]=y/R-xy/R^2+O_{\mathbb P}((|x|+|y|)^3/R^3)$ yields Eqs.~\eqref{eq:phase-delta} and~\eqref{eq:phase-remainder}.  Independence and symmetry eliminate odd moments in the mean-squared error and give Eq.~\eqref{eq:phase-mse}.  The only quadratic contribution to the mean is
$N^{-2}R^{-2}\sum_jc_js_j\sigma^2$; since
$N^{-1}\sum_jc_js_j=\tfrac12\operatorname{Im}(\e^{-2\ii\Psi}Z_2)$, Eq.~\eqref{eq:phase-bias} follows.
\end{proof}

The prefactor and bias are governed by complementary components of the same second harmonic.  Indeed,
\begin{equation}
 q_2=\frac12\left[1+\operatorname{Re}\!\left(\e^{-2\ii\Psi}Z_2\right)\right]
 \geq R^2,
 \label{eq:q2-second-harmonic}
\end{equation}
where the inequality follows from
$N^{-1}\sum_jc_j=R$ and Cauchy--Schwarz.

The familiar $1/(NR^2)$ dependence is therefore a conditional small-error law, not a universal variance of a circular random phase.  Its use requires both small individual angular errors, $\sigma\ll1$, and a small collective perturbation relative to the order parameter,
$\sigma^2\ll NR^2$.  The order-one prefactor $q_2$ depends on the microscopic configuration and the threshold depends on the observation noise and required accuracy.  For a target root-mean-square error $\delta$, a defensible operational condition is
\begin{equation}
 NR^2\gtrsim\frac{\sigma^2q_2}{\delta^2},
 \label{eq:operational-threshold}
\end{equation}
rather than a protocol-independent constant $NR^2>\kappa$.

\subsection{Information retained by the order parameter}

The divergence in Eq.~\eqref{eq:phase-mse} is not necessarily a loss of information in the full phase sample.  To make that distinction identifiable, suppose a reference configuration $\bm\theta^0$ is known up to a rigid rotation and the localized observations are Gaussian,
\begin{equation}
 \widetilde\theta_j=\theta_j^0+\Psi+\epsilon_j,
 \qquad \epsilon_j\sim\mathcal N(0,\sigma^2).
 \label{eq:known-template-model}
\end{equation}

\begin{proposition}[Information and compression efficiency]
\label{prop:fisher-efficiency}
In the known-template model, the Fisher information and Cram\'er--Rao bound are~\cite{Rao1945}
\begin{equation}
 \mathcal I(\Psi)=\frac{N}{\sigma^2},
 \qquad
 \Var(\widehat\Psi)\geq\frac{\sigma^2}{N}.
 \label{eq:fisher-cramer-rao}
\end{equation}
The localized circular estimator
\begin{equation}
 \widehat\Psi_{\rm full}
 =\arg\!\left[\frac1N\sum_j
 \e^{\ii(\widetilde\theta_j-\theta_j^0)}\right]
 \label{eq:template-estimator}
\end{equation}
attains this bound to leading order.  By contrast, compression to the ordinary order parameter has asymptotic efficiency
\begin{equation}
 \eta_{\arg Z}
 :=\frac{\sigma^2/N}{\sigma^2q_2/(NR^2)}
 =\frac{R^2}{q_2}\leq1.
 \label{eq:order-parameter-efficiency}
\end{equation}
\end{proposition}

Thus the $R^{-2}$ amplification quantifies information discarded by the statistic $Z$, not information necessarily absent from the observations.  This conclusion depends essentially on knowing the microscopic template up to rotation.  If every $\theta_j^0$ is an unrestricted nuisance parameter, the separation of template and global rotation is gauge degenerate and $\Psi$ is not identifiable without additional structure.

\subsection{Dynamical diffusion of the neutral phase}

For Eq.~\eqref{eq:stochastic-model} on a symmetric graph, consider a continuous lift of the phases and define
$\bar\theta=N^{-1}\sum_i\theta_i$.  Pairwise coupling cancels exactly in the sum.

\begin{proposition}[Collective zero-mode diffusion]
\label{prop:zero-mode}
For symmetric $A$,
\begin{equation}
 \dd\bar\theta=\sqrt{\frac{2D}{N}}\dd W_{0,t},
 \qquad
 \Var[\bar\theta(t+\Delta t)-\bar\theta(t)]
 =\frac{2D\Delta t}{N}.
 \label{eq:zero-mode-diffusion}
\end{equation}
Writing $\theta_j=\bar\theta+\varphi_j$ with
$\sum_j\varphi_j=0$ and $\eta=\max_j|\varphi_j|\ll1$, direct expansion gives
\begin{equation}
 \e^{-\ii\bar\theta}Z
 =1-\frac{1}{2N}\sum_j\varphi_j^2
 -\frac{\ii}{6N}\sum_j\varphi_j^3+O(\eta^4),
 \label{eq:phase-mean-expansion}
\end{equation}
and hence
$\Psi=\bar\theta-(6N)^{-1}\sum_j\varphi_j^3+O(\eta^5)$.
Thus $\Psi=\bar\theta+O(\eta^3)$, so the global order-parameter phase has the same leading diffusion law.
\end{proposition}

This neutral Brownian motion should not be confused with the ill-conditioning of $\arg Z$ as $R\to0$.  The first is physical diffusion of a symmetry mode; the second is loss of estimability of a polar angle near the origin.

\section{Periodic rings and winding sectors}
\label{sec:ring}

Clock equivalence does not imply global synchronization.  The autonomous flow may have multiple attractors, and a scalar protocol preserves their orbit structure.  A nearest-neighbor ring is the simplest example~\cite{Wiley2006,Delabays2017,Zhang2021,Mihara2022,Ferguson2018,MedvedevTang2015}.

For periodic indexing, define
\begin{equation}
 W(\bm\theta)=\frac{1}{2\pi}\sum_{j=1}^{N}
 \wrap(\theta_{j+1}-\theta_j)\in\mathbb Z.
 \label{eq:winding}
\end{equation}
The $q$-twisted configurations
\begin{equation}
 \theta_j^{(q)}=\phi+\frac{2\pi qj}{N}
 \label{eq:twisted-state}
\end{equation}
are equilibria.  Their linearization is
$-\cos(2\pi q/N)L_{\rm ring}$, hence they are linearly stable modulo the neutral global-phase direction when
\begin{equation}
 |q|<N/4.
 \label{eq:twisted-stability}
\end{equation}
For $q\not\equiv0\pmod N$,
\begin{equation}
 R_q=0,
 \qquad
 C_{b,q}:=\frac1N\sum_j\cos(\theta_{j+1}-\theta_j)
 =\cos(2\pi q/N).
 \label{eq:local-v-global}
\end{equation}
Thus a nonzero-winding state can be almost perfectly ordered locally while lacking a nonzero global order parameter.  This is not a failure of the Laplacian relaxation rate within a sector; it is multistability between sectors with qualitatively different global observables.

For a discrete ring, $W$ can change when an edge increment crosses the branch cut.  Once the flow enters a phase-cohesive region, changes of $W$ cease and the system relaxes toward the corresponding twisted state.  Additive noise restores rare phase slips and, through Eq.~\eqref{eq:stochastic-timechange}, makes their statistics schedule dependent.  In a quasistatic weak-noise approximation, a barrier $\Delta U$ gives the Kramers rate per unit exposure
\begin{align}
 r_s(K)&\asymp A(K)\exp[-K\Delta U/D],\notag\\
 \E N_{\rm slip}&\simeq\int_0^{S_T}r_s(K(t(s)))\dd s,
 \label{eq:kramers-schedule}
\end{align}
which depends on the schedule $K(t(s))$, not only on $S_T$.

\section{Numerical experiments}
\label{sec:numerics}

The simulations are designed to test distinctions established analytically rather than to infer them from an unconstrained parameter collapse.  All random seeds, graph constructors, integrators, and plotting commands are contained in the accompanying script.

\subsection{Matched-exposure protocols}

We use three nonnegative protocols on $0\leq t\leq T$ with the same total exposure $S_T$:
\begin{align}
 K_0(t)&=\bar K,
 \label{eq:kconstant}\\
 K_{\rm F}(t)&=\bar K[1+a\cos(\pi t/T)],
 \label{eq:kfront}\\
 K_{\rm B}(t)&=\bar K[1-a\cos(\pi t/T)],
 \label{eq:kback}
\end{align}
where $\bar K=S_T/T$ and $a=0.9$.  The front-loaded protocol applies most coupling early; the back-loaded protocol applies it late.

For the deterministic test we use one connected Erd\H{o}s--R\'enyi graph with $N=60$, $p=0.10$, i.i.d. uniform initial phases, $T=12$, and $S_T=8$.  The ODE is integrated with DOP853 at relative tolerance $2\times10^{-11}$.  Figure~\ref{fig:clock} shows that $R(t)$ differs strongly in physical time while the curves coincide when parameterized by $S(t)$.  At $T$, all protocols give
$R=0.9975427104$; their phase configurations agree with the constant protocol to rms discrepancies $1.6\times10^{-7}$ and $7.4\times10^{-7}$ for the front- and back-loaded cases, respectively.

\begin{figure}[t]
 \includegraphics[width=\columnwidth]{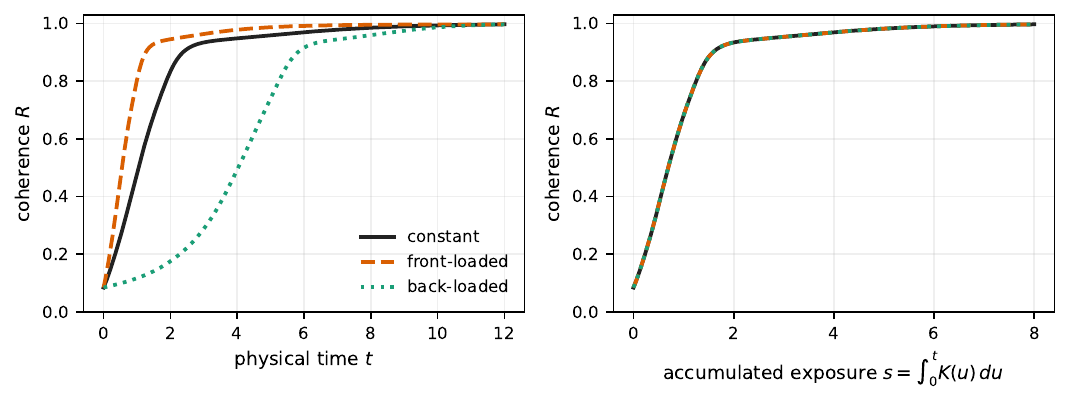}
 \caption{\label{fig:clock}
 Exact clock equivalence.  The three protocols in
 Eqs.~\eqref{eq:kconstant}--\eqref{eq:kback} have the same exposure.  They generate different time traces (left) but the same deterministic orbit when plotted against $S(t)$ (right).}
\end{figure}

\subsection{Reciprocal-kernel response}

We test Theorem~\ref{thm:reciprocal-kernel} with the additive perturbation $\varepsilon\bm G=\varepsilon\bm\omega$, where $\bm\omega$ is one fixed zero-mean, unit-variance frequency vector.  A connected ER graph with $N=28$, $T=6$, and $S_T=4$ is evolved with the three matched-exposure protocols.  The unperturbed orbit is integrated in $s$, and the variational equation in Theorem~\ref{thm:reciprocal-kernel} is solved independently for each reciprocal weight.  We then compare it with finite differences of the full nonlinear heterogeneous system for $3\times10^{-4}\leq\varepsilon\leq10^{-2}$.

Figure~\ref{fig:kernel} shows two distinct facts.  The reciprocal weights have different shapes but equal numerical mass $6.0000002$ or better, as required by Eq.~\eqref{eq:reciprocal-mass}.  Their response gains differ by more than a factor of five, and the nonlinear finite differences converge to the schedule-specific kernel predictions as $\varepsilon\to0$; at the smallest amplitude every relative vector error is below one percent.

\begin{figure}[t]
 \includegraphics[width=\columnwidth]{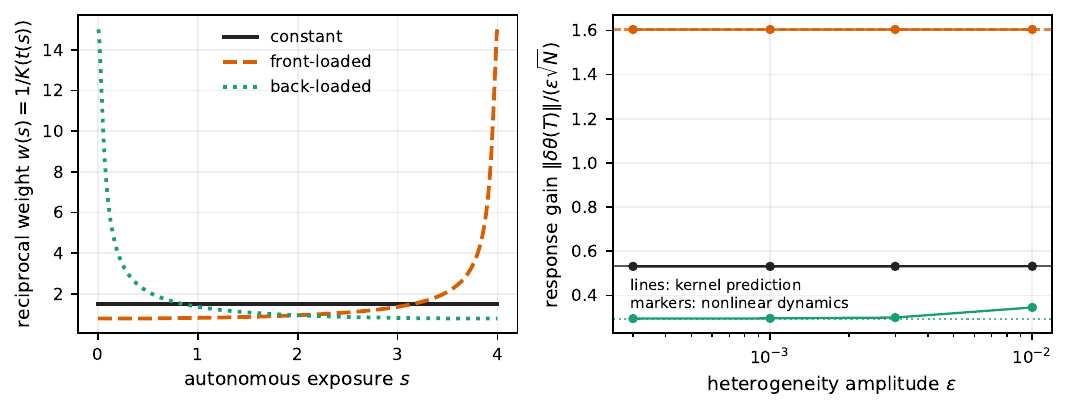}
 \caption{\label{fig:kernel}
 Reciprocal-schedule response to weak frequency heterogeneity.  Left: the three weights $w(s)$ redistribute the same mass $T$ along one autonomous orbit.  Right: nonlinear response gains (markers) converge to the first-order kernel predictions (horizontal lines).}
\end{figure}

\subsection{Spectral exposure}

To test Eq.~\eqref{eq:mode-solution} independently of random-phase nonlinear transients, we initialize the system along the normalized $\lambda_2$ eigenvector with phase rms $0.05$.  We use two connected ER graphs with $N=60$ and $p=0.08,0.14$, and a nearest-neighbor ring with $N=60$.  The autonomous unit-coupling flow is integrated with fourth-order Runge--Kutta and plotted against $\lambda_2s$.

As shown in Fig.~\ref{fig:spectral}, the normalized mode amplitudes follow $\exp(-\lambda_2s)$.  The largest relative deviations over $0\leq\lambda_2s\leq4$ are $1.7\times10^{-3}$, $4.7\times10^{-3}$, and $3.5\times10^{-6}$ for the sparse ER, denser ER, and ring graphs.  These small departures are the expected nonlinear corrections at finite initial amplitude.

\begin{figure}[t]
 \includegraphics[width=0.78\columnwidth]{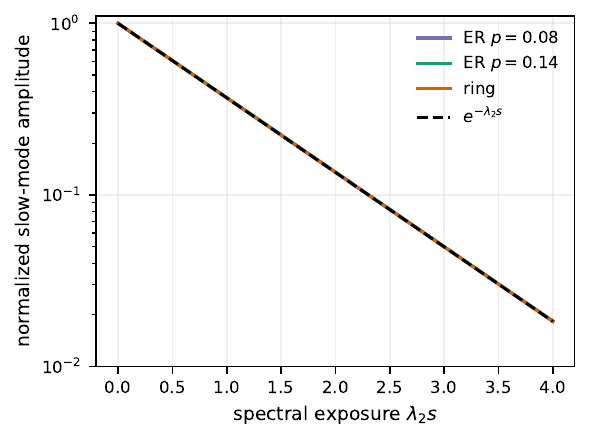}
 \caption{\label{fig:spectral}
 Slow Laplacian modes against spectral exposure.  In the controlled near-coherent regime, different graphs follow the parameter-free prediction $|c_2(s)/c_2(0)|=\exp(-\lambda_2s)$.}
\end{figure}

\subsection{Noise produces a genuine schedule effect}
\label{sec:noise-experiment}

For the stochastic test we use a connected ER graph with $N=40$, $p=0.16$, and measured $\lambda_2=1.245087$.  We take $T=10$, $S_T=7$, and use the same three smooth protocols.  For the bounded-control comparison we add Eq.~\eqref{eq:terminal-bangbang} with $K_{\max}=1.4$, which switches at $t=5$.  The left panel of Fig.~\ref{fig:noise} evaluates the exact variance equation~\eqref{eq:ou-variance} for the slow mode with $D=0.02$ and $v_2(0)=0.20$.  The equal-exposure final variances are
\begin{equation}
 v_2(T)=
 \begin{cases}
 0.02295,&K_0,\\
 0.09285,&K_{\rm F},\\
 0.01213,&K_{\rm B},\\
 0.01147,&K_*.
 \end{cases}
 \label{eq:ou-numerical}
\end{equation}

We test Eq.~\eqref{eq:ou-variance} directly in the full nonlinear network using an ensemble of $512$ near-coherent states with Gaussian slow-mode coefficient of variance $v_2(0)=0.20$.  Each realization is projected onto the normalized $\lambda_2$ eigenvector throughout the stochastic evolution.  The measured final variances are
\begin{equation}
 \widehat v_2(T)=0.0228(15),\ 0.0937(58),\ 0.01265(86),\ 0.01203(73)
 \label{eq:mode-validation}
\end{equation}
for the constant, front-loaded, back-loaded, and bounded-optimal protocols, with standard errors in parentheses.  They differ from the exact linear predictions in Eq.~\eqref{eq:ou-numerical} by $0.6\%$, $0.9\%$, $4.3\%$, and $4.8\%$, respectively, and all four predictions lie within one standard error of the measurement.  Figure~\ref{fig:noise} shows agreement over the complete time evolution and confirms that $K_*$ gives the smallest terminal variance in the admissible class.

We separately integrate the full nonlinear stochastic network from $128$ random initial configurations per point using Euler--Maruyama with $\Delta t=10^{-2}$.  At $D=0$, equal-exposure protocols reach the same state up to discretization error.  At $D=0.03$, the final mean coherences are
\begin{equation}
 R(T)=0.99510(12),\quad0.96853(75),\quad0.997370(65)
 \label{eq:nonlinear-noise-results}
\end{equation}
for constant, front-loaded, and back-loaded schedules, where parentheses denote the standard error in the last quoted digits.  Because the back-loaded protocol has the largest surviving exposure $S_T-S(u)$ at every $u$, it attenuates noise injected throughout the interval most strongly.  The front-loaded protocol instead leaves a weakly coupled terminal interval in which phase fluctuations accumulate.

\begin{figure}[t]
 \includegraphics[width=\columnwidth]{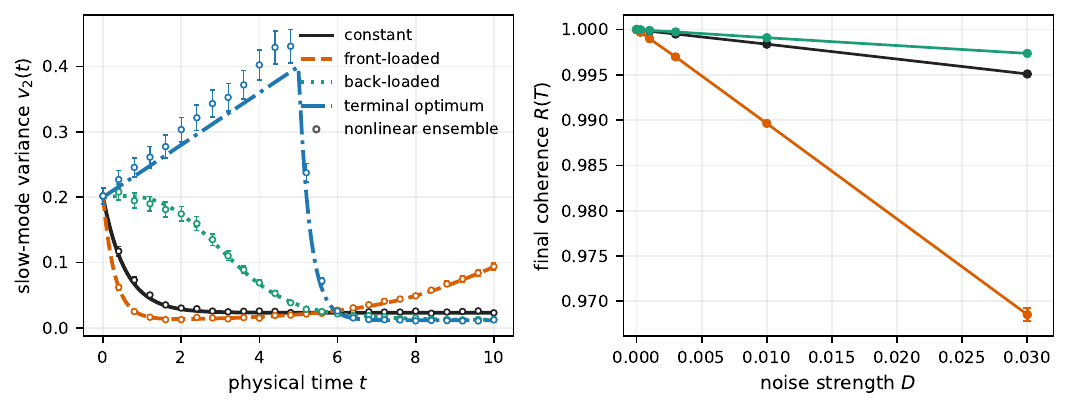}
 \caption{\label{fig:noise}
 Stochastic breaking of exposure equivalence.  Left: exact slow-mode variance from Eq.~\eqref{eq:ou-variance} (curves) and direct projections of $512$ nonlinear-network realizations (open circles).  The blue terminal bang--bang protocol is the bounded optimum.  Right: final coherence for the three smooth schedules from random initial phases versus $D$.  Error bars are standard errors.}
\end{figure}

\subsection{Phase estimability under controlled observation noise}

We test Proposition~\ref{prop:phase-error} using configurations drawn from von Mises distributions with concentrations
$0.3,0.6,1,2,4,8$ and sizes $N=20,50,100,200$.  Each microscopic configuration is held fixed while $320$ independent observation errors of standard deviation $\sigma=0.04$ are applied.  After removing configurations with $R<0.04$, the test contains $670$ configurations.

Figure~\ref{fig:estimability} shows the wrapped phase mean-squared error against $NR^2$.  Rescaling by $NR^2/\sigma^2$ removes the leading size and coherence dependence: the median relative difference from the configuration-dependent prediction $q_2$ is $5.3\%$.  For the known-template estimator in Eq.~\eqref{eq:template-estimator}, the median ratio of the observed MSE to the Cram\'er--Rao value $\sigma^2/N$ is $0.996$.  The measured efficiency of $\arg\widetilde Z$ follows $R^2/q_2$, with median absolute discrepancy $0.011$.

\begin{figure*}[t]
 \includegraphics[width=\textwidth]{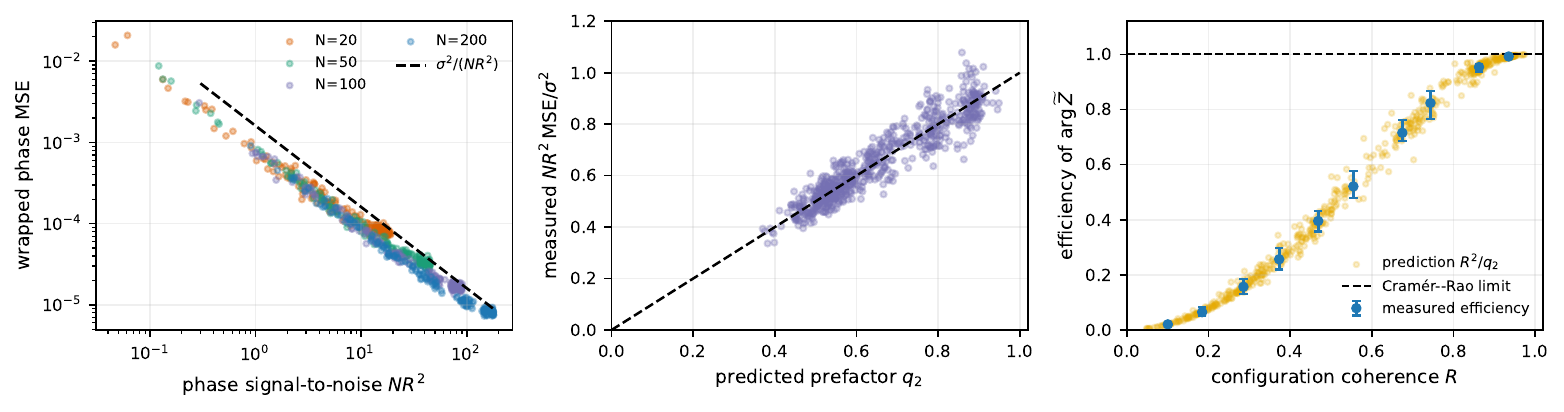}
 \caption{\label{fig:estimability}
 Information retained by the order parameter.  Left: wrapped phase MSE of $\arg\widetilde Z$ versus $NR^2$.  Center: scaled error compared configuration by configuration with the predicted prefactor $q_2$.  Right: predicted efficiency $R^2/q_2$ and binned measurements relative to the template-aware estimator; the dashed line is the Cram\'er--Rao limit.}
\end{figure*}

\subsection{Direct measurement of winding sectors}

We finally simulate $M=400$ nearest-neighbor rings with $N=80$ from i.i.d. uniform phases in the autonomous exposure clock.  The ring gap is
$\lambda_2=2-2\cos(2\pi/N)=0.00616533$.  A second-order Runge--Kutta scheme with $\Delta s=0.04$ is used up to
$\lambda_2s=3.2$.  We record $W$, $R$, and $C_b$ separately rather than attributing all deviations to the small gap.

The sector probabilities settle by $\lambda_2s\simeq0.05$ and remain unchanged over the rest of the simulation: at the final exposure, $24.25\%$ of runs have $W=0$, $39.5\%$ have $|W|=1$, $26.5\%$ have $|W|=2$, and $9.75\%$ have $|W|\geq3$.  Within each sector, local relaxation continues.  At $\lambda_2s=3.2$, the conditional means are
\begin{center}
\begin{tabular}{c@{\quad}ccc}
\toprule
$|W|$ & count & $\E[R\mid |W|]$ & $\E[C_b\mid |W|]$\\
\midrule
0 & 97  & 0.9954 & 0.99997\\
1 & 158 & 0.0637 & 0.99689\\
2 & 106 & 0.00259& 0.98766\\
3 & 28  & 0.00010& 0.97233\\
4 & 10  & $3.2\times10^{-6}$ & 0.95102\\
\bottomrule
\end{tabular}
\end{center}
The single $|W|=5$ realization has $C_b=0.92387$ and
$R\simeq10^{-9}$.  Figure~\ref{fig:ring} displays both the early sector selection and the separation between local and global order.

\begin{figure}[t]
 \includegraphics[width=\columnwidth]{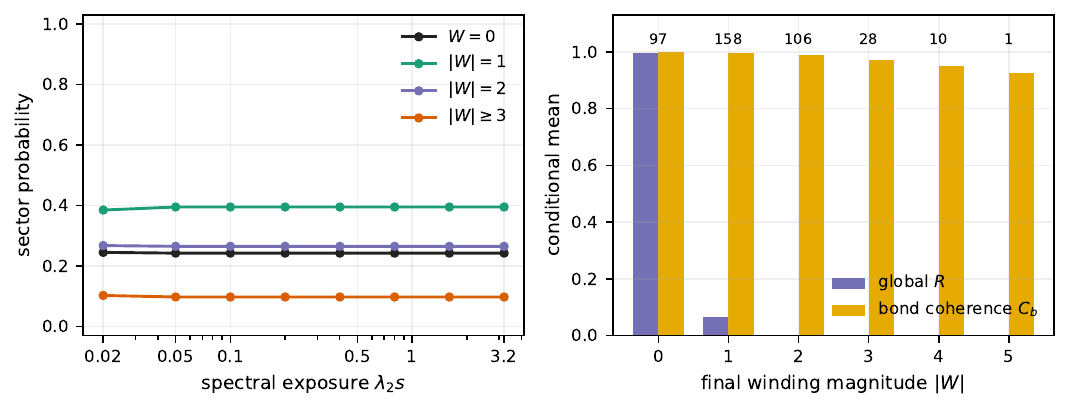}
 \caption{\label{fig:ring}
 Winding sectors on a periodic ring.  Left: sector probabilities stabilize at small spectral exposure.  Right: final global coherence and local bond coherence conditioned on $|W|$; numbers above the bars are realization counts.}
\end{figure}

\section{Discussion}
\label{sec:discussion}

The exact clock reduction changes the interpretation of scalar coupling ramps in the identical overdamped model.  If two deterministic protocols differ only in how the same exposure is distributed in time, their states cannot differ at matched exposure.  A reported rate effect obtained by changing $\tau$ while also changing $\int K\dd t$ does not by itself establish nonadiabaticity.

The exposure clock acts as a gauge redundancy of the deterministic scalar sector: changing $K(t)$ changes the parametrization of an orbit, not the orbit itself.  The results can therefore be summarized by their transformation under this reparametrization.

\begin{table*}[t]
\centering
\caption{\label{tab:classification}Classification under the exposure-clock reparametrization.}
\begin{tabular}{p{0.17\textwidth}p{0.34\textwidth}p{0.39\textwidth}}
\toprule
Class & Mathematical signature & Representative quantities or mechanisms\\
\midrule
Invariant & Depends on the autonomous orbit at $S$ & Deterministic state, spectral exposure $\lambda_\alpha S$, basin and winding sector\\
Reciprocal weighted & Linear functional or covariance integral with $w(s)=1/K(t(s))$ & Weak frequency heterogeneity, autonomous pinning or drift, laboratory-time noise\\
Outside scalar reparametrization & Additional dependence on $t(s)$, derivatives of the clock, or noncommuting generators & Explicitly timed forcing, inertia, delays, time-dependent graph structure\\
\bottomrule
\end{tabular}
\end{table*}

This does not make time-dependent coupling physically empty.  It identifies what must be specified.  If the experimental resource is physical time rather than exposure, different schedules reach the same orbit points at different times.  Finite-exposure protocols can arrest evolution before an attractor is reached.  Additive perturbations redistribute their action along the orbit through $w(s)$, and topology can select an invariant attractor with strong local but vanishing global order.  These mechanisms require different observables.

The stochastic calculation also clarifies when an adiabaticity criterion is legitimate.  In a linear mode, the instantaneous equilibrium variance $D/[K(t)\lambda_\alpha]$ moves as $K$ changes.  Comparing its relaxation time with the protocol time measures covariance tracking.  By contrast, the deterministic fixed point and its basin are independent of a positive scalar $K$; there is no moving deterministic target to track.

The phase analysis similarly separates geometry, information, and dynamics.  The argument of a complex order parameter is ill-conditioned near $R=0$, but the full phase sample retains Fisher information when a microscopic template is known.  The efficiency $R^2/q_2$ quantifies loss caused by compression to $Z$.  Once coherence is present, the neutral collective phase still undergoes Brownian diffusion at rate $2D/N$.  A complete operational statement must therefore identify the observation model, statistic, and dynamical quantity.

\subsection{Boundary of the Kibble--Zurek interpretation}

The conventional Kibble--Zurek mechanism concerns a finite-rate passage through a critical point, where a diverging relaxation time prevents the state from tracking the changing equilibrium structure~\cite{Kibble1976,Zurek1985}.  For $K>0$ in the exact-clock model, changing $K$ neither moves equilibria nor changes their stability signatures relative to the autonomous clock.  There is therefore no $K$-dependent deterministic critical manifold to cross, and conventional Kibble--Zurek freeze-out is unavailable in this sector.

The same statement identifies where that interpretation can become valid.  Frequency disorder creates a $K$-dependent entrainment threshold, and laboratory-time noise creates a $K$-dependent invariant measure; both terms break clock equivalence.  Experiments on synchronization fronts that exhibit defect trapping cross an actual secondary bifurcation and belong to this clock-breaking class~\cite{Miranda2013}.  Quantifying how a rate-controlled defect law opens from the exact-clock limit requires a separate finite-size study: on a local ring, the locking threshold and winding statistics are not interchangeable with the mean-field critical point.  The present result supplies the diagnostic boundary rather than claiming that crossover without the required scaling analysis.

Finally, $\lambda_2$ remains important but has a delimited role.  It controls the asymptotic decay of the slowest linear mode within a coherent basin and enters mode-dependent protocol objectives.  It does not encode the entire nonlinear transient, the basin volume of twisted states, or phase-slip kinetics between sectors.  Comparing graph families solely through $\lambda_2$ can therefore obscure rather than reveal topological obstruction.

\section{Conclusions}

The question motivating this work was whether changing the global coupling strength produces a genuinely driven synchronization problem, or merely changes the speed at which an otherwise autonomous dynamics is traversed.  

For identical overdamped oscillators on a fixed graph, the answer is exact: a scalar protocol $K(t)$ acts only as an internal clock.  Once time is replaced by the accumulated exposure $S(t)=\int_0^tK(u)\dd u$, all positive protocols generate the same deterministic orbit.  Equal-exposure protocols therefore cannot produce distinct deterministic freeze-out states.  

An integrable decay of the coupling can nevertheless arrest the evolution, but it does so because the system receives only a finite amount of interaction time, not because it fails to follow a moving equilibrium.  In the coherent regime, the combinations $\lambda_\alpha S$ then describe how far each graph mode has progressed along this common orbit.

This equivalence is useful less as an endpoint than as a reference against which genuine protocol dependence can be identified.  A perturbation that is not multiplied by $K(t)$ reappears in the exposure clock with a reciprocal schedule weight.  Consequently, weak frequency heterogeneity, pinning, or an external drift acquires a universal first-order response kernel weighted by $1/K$, while laboratory-time noise produces the corresponding reciprocal-weighted covariance.  

Once a perturbation such as noise or a small frequency mismatch is added, two schedules with the same total coupling can produce different outcomes because the timing of the coupling now matters. For noise, the reason is simple. A front-loaded protocol synchronizes the system strongly at the beginning but leaves a weakly coupled interval during which newly generated fluctuations can accumulate. A back-loaded protocol instead applies the strongest relaxation near the end and therefore suppresses more of the noise accumulated during the evolution. If both the total coupling and the maximum allowed coupling strength are fixed, the smallest final fluctuations are consequently obtained by postponing the coupling and applying it at its maximum value over the final part of the protocol. This result depends on what is held fixed: with a different cost, such as one that penalizes large coupling strengths, the optimal schedule need not have this form and may depend on the graph spectrum. Inertia and changes in the network structure are more fundamental departures, because they modify the dynamics itself rather than merely changing when an external perturbation acts.

The same separation is needed at the level of observables.  A small value of the global order parameter does not have a unique interpretation.  The estimator $\arg Z$ becomes inefficient as $R\to0$, even though the full collection of phase observations can retain finite Fisher information about a rigid rotation when the microscopic template is known.  On a periodic graph, the ambiguity is dynamical as well as statistical: stable winding states may possess strong local coherence while their global order parameter vanishes.  Thus poor global phase estimability, weak local synchronization, and topological cancellation are distinct phenomena, although all can appear simply as a small value of $R$.

These results suggest a practical order of analysis for time-dependent synchronization.  The first test should be whether protocols with the same accumulated exposure generate the same deterministic state.  If they do, apparent rate effects should be interpreted relative to the exposure clock.  If they do not, the term that breaks the equivalence and the observable sensitive to it should be identified explicitly.  Only when additional physics creates a $K$-dependent equilibrium, invariant measure, or critical structure is there a changing target whose loss of tracking can support an adiabatic or Kibble--Zurek interpretation.  In this sense, clock equivalence is not merely a restriction of the identical-oscillator model: it supplies the baseline from which the onset of genuine rate dependence can be defined and measured.

\section*{Acknowledgments}

The work
of V.S. is supported by the Spanish grants PID2023-148162NB-C21,
and CEX2023-001292-S.

\section*{Data and code availability}
The complete simulation and figure-generation script, fixed seeds, and numerical summary are supplied with this manuscript in this GitHub repository  \url{https://github.com/vsanz/scalar-coupling-protocols}. An archival
DOI will be added upon submission.

\appendix

\section{Numerical convergence and reproducibility}
\label{app:numerics}

All simulations use the fixed master seed $260305668$.  ER graphs are regenerated until connected.  Their Laplacian spectra are computed by dense symmetric diagonalization.  Deterministic matched-protocol trajectories and the reciprocal-kernel test use adaptive DOP853; the latter integrates the autonomous orbit and its variational equation independently before comparison with the full heterogeneous dynamics.  The stochastic network uses Euler--Maruyama with $\Delta t=10^{-2}$, including a grid-aligned representation of the terminal bang--bang switch.  The ring calculation uses a midpoint second-order Runge--Kutta method in the exposure clock with $\Delta s=0.04$.  The near-coherent spectral-mode tests use fourth-order Runge--Kutta with step sizes selected from the largest Laplacian eigenvalue.

The code writes all plotted numerical values to a machine-readable JSON file.  Repeating the clock-equivalence experiment at half the maximum DOP853 step changes the final $R$ below $10^{-11}$.  Repeating the ring experiment with $\Delta s=0.02$ for a fixed subset of $100$ initial conditions leaves all final winding labels unchanged and changes conditional bond coherences below $3\times10^{-5}$.  A coupled-Brownian-path check of the back-loaded stochastic mode experiment, using a separate ensemble of $128$ coupled paths rather than the $512$ realizations in Sec.~\ref{sec:noise-experiment}, gives final sample variances $0.013004$ and $0.012966$ at $\Delta t=0.01$ and $0.005$, a relative step-size difference of $2.9\times10^{-3}$.

\section{Why raw phase variance is not a global observable}

For a circular random variable $\Psi$, the linear variance depends on the choice of branch whenever probability mass approaches $\pm\pi$.  Two standard branch-independent summaries are the mean resultant length
\begin{equation}
 \rho_\Psi=|\E\e^{\ii\Psi}|
 \label{eq:circular-resultant}
\end{equation}
and circular variance $1-\rho_\Psi$.  For a gauge-invariant temporal diagnostic one may instead use phase increments
\begin{equation}
 \Delta\Psi(t;t_0)=\wrap[\Psi(t)-\Psi(t_0)]
 \label{eq:phase-increments}
\end{equation}
and report either $\E[(\Delta\Psi)^2]$ while increments remain localized or
$1-|\E\e^{\ii\Delta\Psi}|$ globally.  Proposition~\ref{prop:phase-error} concerns the first, local regime under a specified observation model.

\begingroup
\small
\setlength{\bibsep}{0pt plus 0.2ex}
\bibliographystyle{unsrtnat}
\bibliography{refs}
\endgroup

\end{document}